\documentclass[12pt]{article}

\usepackage{amsmath,amssymb,amsthm}

\usepackage{tikz}
\usetikzlibrary{calc,intersections, through,arrows,decorations.markings,quotes}

\newcommand{\row}{\mathrm{row}}
\newcommand{\qbinom}[2]{\genfrac{[}{]}{0pt}{}{\,#1\,}{#2}_q}
\newcommand{\Prob}{\mathrm{P}}
\newcommand{\rk}{\mathrm{rk}}
\newcommand{\bE}{\mathbf{E}}

\newcommand{\bX}{\mathbf{X}}
\newcommand{\bY}{\mathbf{Y}}

\newtheorem{theorem}{Theorem}
\newtheorem{lemma}[theorem]{Lemma}
\newtheorem{corollary}[theorem]{Corollary}

\begin{document}
\title{The Capacity of a Finite Field Matrix Channel}
\author{Simon R. Blackburn\thanks{Department of Mathematics,
Royal Holloway University of London, Egham, Surrey TW20 0EX, United Kingdom. {\tt s.backburn@rhul.ac.uk}}\,\, and
Jessica Claridge\thanks{School of Mathematics, Statistics and Actuarial Science, University of Essex,
Wivenhoe Park, Colchester, Essex, CO4 3SQ, United Kingdom. {\tt jessica.claridge@essex.ac.uk}\newline\indent
The authors would like to thank the London Mathematical Society, who funded their collaboration via a Scheme~4 Grant reference~42124, and the reviewers for their very useful comments on this article.
}
}
\maketitle

\begin{abstract}

The Additive-Multiplicative Matrix Channel (AMMC) was introduced by Silva, Kschischang and K\"otter in 2010 to model data transmission using random linear network coding. The input and output of the channel are $n\times m$ matrices over a finite field $\mathbb{F}_q$. When the matrix $X$ is input, the channel outputs $Y=A(X+W)$ where $A$ is a uniformly chosen $n\times n$ invertible matrix over $\mathbb{F}_q$ and where $W$ is a uniformly chosen $n\times m$ matrix over $\mathbb{F}_q$ of rank $t$.

Silva \emph{et al.} considered the case when $2n\leq m$. They determined the asymptotic capacity of the AMMC when $t$, $n$ and $m$ are fixed and $q\rightarrow\infty$. They also determined the leading term of the capacity when $q$ is fixed, and $t$, $n$ and $m$ grow linearly. We generalise these results, showing that the condition $2n\geq m$ can be removed. (Our formula for the capacity falls into two cases, one of which generalises the $2n\geq m$ case.) We also improve the error term in the case when $q$ is fixed.
\end{abstract}

\paragraph{Index terms:} Random network coding; matrix channels; channel capacity.

\section{Introduction}
\label{sec:introduction}

The \emph{additive-multiplicative matrix channel} (AMMC) was introduced by Silva, Kschischang and K\"otter~\cite{SilvaKoetter} for applications to network coding, and is defined as follows. Let $m$, $n$ and $t$ be integers where $m$ and $n$ are positive and  $0\leq t\leq \min\{n,m\}$. The input and output of the channel are $n\times m$ matrices over a finite field~$\mathbb{F}_q$. (All matrices in this paper will have entries in $\mathbb{F}_q$.) The AMMC takes a matrix $X$ as input, and outputs the matrix $Y$ defined by:
\begin{equation}
Y=A(X+W),
\end{equation}
where $A$ and $W$ are chosen uniformly and independently at random, so that $A$ is an $n\times n$ invertible matrix, and $W$ is an $n\times m$ matrix of rank $t$.

Before discussing results on the AMMC channel, we provide some context.  Ahlswede, Cai, Li and Yeung~\cite{AhlswedeCai} were the first to realise that allowing intermediate nodes of a network to compute with and modify packets of data can significantly increase the rate of information flow through the network. For multicast problems, Li, Yeung and Cai~\cite{LiYeung} showed that linear network coding (where data packets are regarded as vectors over some finite field $\mathbb{F}_q$, and intermediate nodes in the network forward linear combinations of the packets they receive) is sufficient to maximise information flow. Moreover, \emph{random linear network coding} (RLNC), where linear combinations are chosen at random, achieves capacity with probability exponentially approaching $1$ with the code length~\cite{HoMedard}. See~\cite{BassoliMarques} for a survey of network coding.

Consider the (unrealistically small) example given in Figure~\ref{fig:RNC}. In RLNC the source injects packets into the network; these packets can be thought of as vectors of length $m$ with entries in a finite field~$\mathbb{F}_q$ (where $q$ is a prime power). In practice $m$ and $q$ are taken to be large, but in our example $m=6$ and $q=2$. The packets flow through a network of unknown topology to a sink node. Each intermediate node forwards packets that are random linear combinations of the packets it has received. A sink node attempts to reconstruct the message from the received packets.  To relate our example to the AMMC, we define matrices $X$ and $Y$ whose rows correspond to the packets transmitted by the source and received by the sink, respectively. The matrix $A$ corresponds to the linear combinations of packets computed by the intermediate nodes and the matrix $W$ corresponds to errors introduced during transmission, and so the model assumes exactly~$t$ linearly independent random errors have been introduced. In the example of Figure~\ref{fig:RNC}, the (rank $3$) input matrix $X$ and the output matrix $Y$ are
\[
X=\begin{pmatrix}
1&0&0&0&0&0\\
0&1&0&0&0&0\\
0&0&1&0&0&0\\
1&0&1&0&0&0\\
1&1&0&0&0&0\\
1&1&1&0&0&0\\
\end{pmatrix}\text{ and }
Y=\begin{pmatrix}
0&0&0&0&0&0\\
0&0&0&0&1&1\\
1&0&0&0&0&0\\
0&1&0&0&1&1\\
0&0&0&0&1&1\\
0&1&1&0&1&1\\
\end{pmatrix}.
\]
Here $Y=A(X+W)$ where
\begin{align*}
A&=
\begin{pmatrix}
1& 0& 1& 0& 1& 0\\
0& 1& 0& 0& 1& 1\\
1& 0& 1& 1& 1& 0\\
0& 0& 0& 0&1& 1\\
0& 0& 1& 1& 0& 1\\
0& 1& 0& 0& 0& 1
\end{pmatrix}
\begin{pmatrix}
1& 1& 0& 0& 0& 0\\
0&1& 0& 0& 0& 0\\
0& 0& 0& 1& 0& 0\\
0& 0& 1&1& 0& 0\\
0& 0& 0& 0& 1& 0\\
0& 0& 0& 0& 0& 1
\end{pmatrix}
=
\begin{pmatrix}
1& 1& 0& 1& 1& 0\\
0&1& 0& 0& 1& 1\\
1& 1& 1& 0& 1& 0\\
0& 0& 0&0& 1& 1\\
0& 0& 1& 0& 0& 1\\
0& 1& 0& 0& 0& 1
\end{pmatrix}\text{ and}\\
W&=\begin{pmatrix}
1& 1& 0& 0& 0& 0\\
0&1& 0& 0& 0& 0\\
0& 0& 0& 1& 0& 0\\
0& 0& 1&1& 0& 0\\
0& 0& 0& 0& 1& 0\\
0& 0& 0& 0& 0& 1
\end{pmatrix}^{\!\!-1}
\begin{pmatrix}
0&0&0&0&0&0\\
0&0&0&0&0&0\\
0&0&0&0&0&0\\
0&0&0&0&0&0\\
1&0&1&0&0&0\\
1&1&0&0&1&1
\end{pmatrix}
=
\begin{pmatrix}
0& 0& 0& 0& 0& 0\\
0& 0& 0& 0& 0& 0\\
0& 0& 0& 0& 0& 0\\
0& 0& 0& 0& 0& 0\\
1& 0& 1& 0& 0& 0\\
0& 1& 0& 0& 1& 1
\end{pmatrix}.
\end{align*}

\begin{figure}
\label{fig:RNC}
\begin{tikzpicture}[fill=gray!50, scale=0.5,>=angle 45,shorten
  >=5,shorten <=5,
vertex/.style={circle,inner sep=2,fill=black,draw}]

\node at (0,0) [vertex, name=X,label=west:$X$]{};
\node at (5,5) [vertex, name=U1,label=north:$U_1$]{};
\node at (5,0) [vertex, name=U2,label=south:$U_2$]{};
\node at (5,-5) [vertex, name=U3,label=north:$U_3$]{};
\node at (10,0) [vertex, name=V,label=south east:$V$]{};
\node at (18,0) [vertex, name=Y,label=east:$Y$]{};

\draw (X) edge ["$x_1 \,x_2$",->,bend left=30] (U1);
\draw (X) edge ["$x_3 \,x_3$",->] (U2);
\draw (X) edge ["$x_5 \,x_6$",->,bend right=30,swap] (U3);

\draw (U1) edge ["$u_1 \,u_2$", ->, bend left=30] (V);
\draw (U2) edge ["$u_3 \,u_4$", ->] (V);
\draw (U3) edge ["$u_5 \,u_6$", ->, swap, bend right=30] (V);

\draw (V) edge ["$y_1 \,y_2 \,y_3 \,y_4 \,y_5 \,y_6$",->] (Y);
\end{tikzpicture}
\caption{An example of random network coding. Here we are using binary vectors of length $6$. Source node $X$ sends (say)
$x_1=(100000), x_2=(010000), x_3=(001000), x_4=(101000), x_5=(110000)$, and $x_6=(111000)$ to nodes $U_i$ as shown.\newline
Nodes $U_1$ and $U_2$ send linear combinations of received vectors to $V$, where $u_1=x_1+x_2$, $u_2=x_2$, $u_3=x_4$ and $u_4=x_3+x_4$. The node $U_3$ is faulty, and sends random vectors $u_5=x_5+(101000)$ and $u_6=x_6+(110011)$ to $V$.\newline
Finally, node $V$ sends the vectors $y_i$ to the sink node $Y$, where $y_1=u_1+u_3+u_5$, $y_2=u_2+u_5+u_6$, $y_3=u_1+u_3+u_4+u_5$, $y_4=u_5+u_6$, $y_5=u_3+u_4+u_6$ and $y_6=u_2+u_6$.}
\end{figure}

Alternative finite-field matrix channels, related to the AMMC, have been proposed to model RLNC. N\'obrega, Silva and Uch\^oa-Filho~\cite{NobregaSilva} considered the multiplicative matrix channel (MMC)
\[
Y=AX
\] 
where now $Y$ is a $n\times m$ matrix, $X$ is a $k\times m$ matrix, and the transfer matrix $A$  is a $n \times k$ matrix chosen according to a distribution that first picks the rank of $A$ according to some distribution on ranks, and then choses $A$ uniformly from all matrices of that rank. Such a distribution we call \textit{Uniform Given Rank (UGR)}. This channel, along with other MMC variants with alternative assumptions on the transfer matrix~\cite{SiavoshaniMohajer},~\cite{YangHo}, have been used to model RLNC when no erroneous packets are injected into the network, but there may be link erasures. N\'obrega, Feng, Silva and Uch\^oa-Filho~\cite{NobregaFeng} and Feng, N\'obrega, Kschischang and Silva~\cite{FengNobrega}, considered the MMC and AMMC, respectively, over finite chain rings, thus moving away from the finite-field case. Finally, Blackburn and Claridge~\cite{BlackburnClaridge} considered a generalisation of the AMMC where the rank of the error matrix $W$ varies, allowing the modelling of different error patterns; Silva et al.~\cite{SilvaKoetter} also discuss such channels. N\'obrega, Silva and Uch\^oa-Filho~\cite{NobregaSilva}, and Blackburn and Claridge~\cite{BlackburnClaridge} both show that for the channels they consider, there always exists a UGR input distribution that achieves capacity. This reduces the problem of finding capacity achieving input distributions to finding a `good' distribution on ranks. 

What is the capacity of the AMMC? As is traditional for channels using matrices over a finite field $\mathbb{F}_q$, we take logarithms to the base $q$ when measuring capacity. (So to compute the capacity in bits, the more usual measure, we must multiply by $\log_2 q$.) Silva et al.\  gave good upper and lower bounds for the capacity of the AMMC when $2n\leq m$. In particular, when $n$, $m$ and $t$ are fixed with $2n\leq m$, and when $q\rightarrow \infty$, their paper establishes~\cite[Corollary~10]{SilvaKoetter} that the capacity tends to $(m-n)(n-t)$ when logarithms are taken to the base $q$. Moreover, Silva et al.\  also provide the leading term for the capacity (of the same form) when $q$ is fixed and $m$, $n$ and $t$ grow linearly.
Our paper aims to generalise these results, by dropping the restriction that $2n\leq m$. We will show (in both asymptotic regimes considered in~\cite{SilvaKoetter}) that the formula for the capacity involves two cases, depending on whether $2n\leq m+t$ or not. Indeed, we show that the capacity is approximately $(m-n)(n-t)$ when $2n\leq m+t$ (generalising the situation considered by Silva et al.) and the capacity is approximately $(m-t)^2/4$ when $2n> m+t$. See Theorems~\ref{thm:AMMC_capacity} and~\ref{thm:fixed_q_AMMC} below for precise statements. (Note that when $m+t\geq 2n$ we must have $m\geq 2n-t=n+(n-t)\geq n$, so the formula for the capacity is never negative!) We mention that the error term in Theorem~\ref{thm:fixed_q_AMMC} is a significant improvement on the corresponding statement in~\cite{SilvaKoetter}. These results confirm a conjectured formula for the channel capacity due to Blackburn and Claridge~\cite[Section~VIII]{BlackburnClaridge}.

Our proofs of Theorems~\ref{thm:AMMC_capacity} and~\ref{thm:fixed_q_AMMC} reduce the computation of the AMMC capacity to finding the capacity of a channel we call the $k$-AMMC, namely the AMMC with the extra restriction that input matrices are restricted to have rank $k$. We think that the $k$-AMMC is a natural channel model, of independent interest. We show in Lemma~\ref{lem:symmetric} that a capacity achieving input distribution of the $k$-AMMC can be taken to be uniform, and use this and the fact that the capacity of the AMMC can be achieved when the input distribution is UGR to establish the capacity of the AMMC.

A natural, and practical, variant of the AMMC defines $W$ to be a matrix of rank at most $t$, rather than exactly $t$. This was first considered in~\cite{SilvaKoetter}, where it is pointed out that the capacity is reduced by at most $\log_q (t+1)$ in this context. We note (as in~\cite{SilvaKoetter}) that this term is negligible in both asymptotic regimes we consider. Moreover, the standard coding schemes for the AMMC (using techniques from subspace coding) are able to cope with error matrices of rank less than $t$ without modification. 

The structure of the paper is as follows. In Section~\ref{sec:overview} we present an overview of our approach, giving intuition behind the expected capacity formula, and briefly describe the relationship between the $k$-AMMC and the subspace-based operator channel of K\"otter and Kschischang~\cite{KoetterKschischang} which has become the standard channel for network coding applications. In Section~\ref{sec:preliminaries} we provide some combinatorial bounds from linear algebra that we need. In Section~\ref{sec:k_ammc} we establish the capacity of the $k$-AMMC when $q\rightarrow\infty$ with other parameters fixed. In Section~\ref{sec:AMMC} we use the results from Section~\ref{sec:k_ammc} to establish the capacity of the AMMC in this asymptotic regime. In Section~\ref{sec:fixed_field}, we show how the results in Sections~\ref{sec:k_ammc} and~\ref{sec:AMMC} may be modified to provide the leading term of the capacity when $q$ is fixed and when $n$, $m$ and $t$ grow linearly. Finally, Section~\ref{sec:coding_scheme} describes a coding scheme for the channel which achieves capacity in both the asymptotic regimes considered.

\section{Overview}
\label{sec:overview}

We give an intuitive (and non-rigorous) justification for the AMMC capacity, as a guide to the proofs in the sections which follow. 

We first consider the $k$-AMMC, so the input to the channel is constrained to matrices of rank $k$. We will show that the capacity achieving input distribution is uniform on such matrices. Note that multiplying by $A$ replaces the matrix $X+W$ with a random matrix having the same row space. Moreover, we can show that the output distribution of the AMMC only depends on the row space of $X$, $\row(X)$, rather than $X$ itself. So the channel can be thought of as having row spaces (of $X$ and $X+W$ respectively) as input and output: the receiver is given $\row(Y)=\row(X+W)$, and must deduce information about $\row(X)$ from this. (But note that we cannot avoid considering matrices entirely, as the noise of the channel is defined using matrix addition rather than an operation on subspaces.)

Now $\row(Y)\subseteq\row(X)+\row(W)$. We would expect $\row(X)$ and $\row(W)$ to have no special relationship to each other (as $\row(W)$ is uniformly and independently chosen of dimension~$t$), so we expect the dimension of $\row(X)+\row(W)$ to be close to $\min\{m,k+t\}$. Hence $\rk(Y)=\dim(\row(Y))\leq \min\{m,k+t\}$. Since $\row(Y)$ is the row space of an $n\times m$ matrix, we must have $\rk(Y)\leq\min\{n,m\}$. We show that $\rk(Y)$ is usually close to being as large as possible given these constraints. So we show that with high probability,
\[
\rk(Y)\approx \min\{m,n,k+t\}=\min\{n,k+t-\delta\},
\]
where $\delta=\max\{0,k+t-m\}$.
We will also determine the approximate dimension of $\row(X)\cap\row(Y)$:
\[
\dim(\row(X)\cap\row(Y))\approx \min\{k,n-(t-\delta)\}
\]
with high probability.
The proof that $\rk(Y)$ and $\dim(\row(X)\cap\row(Y))$ can be approximately determined (see Lemmas~\ref{lem:generic_dimensions} and~\ref{lem:qfixed_generic_dimensions}) is the most difficult part of the argument. 

The receiver initially knows nothing of $\row(X)$, so the entropy of $\row(X)$ is equal to the logarithm to the base $q$ of the number of dimension $k$ subspaces of $\mathbb{F}_q^m$. Once the receiver obtains $Y$, the entropy of $\row(X)$ is reduced, as $\row(X)$ is likely to have an intersection of dimension approximately $\min\{k,n-(t-\delta)\}$ with the (known) subspace $\row(Y)$ of dimension approximately $\min\{n,k+t-\delta\}$. Counting these subspaces allows us to determine the capacity of the $k$-AMMC, by taking logarithms to the base $q$.

We then show that the capacity of the AMMC is essentially the maximum of the capacities of the $k$-AMMC as $k$ varies (because there are few possibilities for $k$). When $k$ is large (more precisely, when $k+t> \min\{n,m\}$) the receiver does not obtain much information about $\row(X)$, as this subspace does not normally intersect $\row(Y)$ in a space of large dimension. We show that it is never advantageous to choose a large value of $k$. When $k+t\leq \min\{m,n\}$, we have that $\rk(X+W)= k+t$ with non-trivial probability, and $\row(Y)=\row(X+W)=\row(X)+\row(W)$. Thus the receiver is given a subspace $\row(Y)$ of dimension $k+t$, and knows that $\row(X)$ lies in this subspace. Our formula for the capacity of the channel is approximately $(m-t-k)k$ in this case. The maximum value of this formula is at $k=(m-t)/2$, which always lies in the range $0\leq k\leq m-t$. So the best value for $k$ depends on whether $(m-t)/2\leq n-t$ (when we take $k\approx (m-t)/2$) or $(m-t)/2> n-t$ (when we take $k=n-t$). This is why the formula for the capacity of the AMMC splits into two cases.

The \emph{subspace distance} $d(U,V)$ between two subspaces $U$ and $V$ is defined to be $\dim(U)+\dim(V)-2\dim(U\cap V)$. This notion was introduced by K\"otter and Kschischang~\cite{KoetterKschischang} in their famous 2008 paper on random network coding, and channels with transition probabilities based on subspace distance have been much studied. We remark that the transition probabilities of the AMMC do not depend only on the subspace distance between the row spaces of input and output matrices of the channel, because we would expect an output matrix $Y$ to have larger rank than the input $X$, and there are matrices $Y'$ of smaller rank such that $d(X,Y)=d(X,Y')$. Nevertheless, our results on the expected values of $\rk(Y)$ and $\dim (X\cap Y)$ above can be interpreted as saying that the subspace distance between the row spaces of input and output matrices in the $k$-AMMC is (approximately) constant with high probability, in the two asymptotic regimes we consider. So the $k$-AMMC may be asymptotically approximated by a channel where the subspace distance between input and output row spaces is bounded. (More precisely, setting $\delta=\max\{0,k+t-m\}$, the $k$-AMMC may be asymptotically approximated by an operator channel~\cite{KoetterKschischang} with $\rho$ erasures and $\tau$ errors, where $\rho=k-\min\{k,n-(t-\delta)\}$ and $\tau=\min\{n,k+t-\delta\}-\min\{k,n-(t-\delta)\}$.) This allows subspace decoding techniques to be used in these asymptotic regimes. (See Kschischang~\cite{Kschischang} or Kurz~\cite{Kurz} for recent surveys of subspace codes.) Moreover, because the likely subspace distance between input and output decreases (possibly not strictly) as $t$ decreases, subspace decoding can be used for relatives of the $k$-AMMC where the rank of $t$ varies, with an upper bound $t$: this is, of course, a very natural class of channels in network coding applications.

\section{Preliminaries}
\label{sec:preliminaries}

This section contains various combinatorial results from linear algebra, all well-known, which we will use in the following sections.

The \textit{Gaussian binomial coefficient}, or \emph{$q$-binomial coefficient}, is written $\qbinom{a}{b}$ and defined to be the number of $b$-dimensional subspaces of an
$a$-dimensional space over $\mathbb{F}_q$. A formula for the Gaussian binomial coefficient is (see, for example, Cameron~\cite[\S 9.2]{Cameron})
\[ 
\qbinom{a}{b}= 
\left\{\begin{array}{ll}
\prod_{i=0}^{b-1} \frac{(q^a-q^i)}{(q^b-q^i)} & \text{for } b \leq a, \\
0 & \text{for } b > a. 
\end{array}\right.
\]

The following lemma is, for example, a consequence of the statement and proof of  K\"otter and Kschischang~\cite[Lemma 4]{KoetterKschischang}).

\begin{lemma}
\label{lem:qbinom}
Let $a$ and $b$ be non-negative integers, with $b\leq a$. For any non-trivial prime power $q$,
\[
\qbinom{a}{b}= f(q)\,q^{(a-b)b}
\]
where 
\[
1< f(q)=\prod_{i=0}^{b-1} \left(\frac{1-q^{-a+i}}{1-q^{-b+i}}\right)<\frac{1}{\prod_{j=1}^b(1-q^{-j})}<4.
\]
In particular, the following inequalities all hold:
\begin{itemize}
\item For any $q$, $a$ and $b$,
\begin{equation}
\label{eqn:general_qbinomial}
q^{(a-b)b}<\qbinom{a}{b}\leq q^{(a-b)b+\log_q 4}.
\end{equation}
\item When $q\rightarrow\infty$ with $a$ and $b$ fixed,
\begin{equation}
\label{eqn:q_big_qbinomial}
\qbinom{a}{b}\,/\,q^{(a-b)b}\rightarrow 1.
\end{equation}
\item When $q\rightarrow\infty$ with $a$ and $b$ fixed,
\begin{equation}
\label{eqn:log_qbinomial}
\log_q \qbinom{a}{b}\rightarrow (a-b)b.
\end{equation}
\end{itemize}
\end{lemma}

\begin{corollary}
\label{cor:span}
Let $a$ and $b$ be fixed non-negative integers, and let $q$ be a non-trivial prime power. Let $U$ be a space of dimension $a$ over $\mathbb{F}_q$. When vectors $x_1,x_2,\ldots,x_b\in U$ are uniformly and independently chosen,
\[
\Prob\big(\dim\langle x_1,x_2,\ldots,x_b\rangle=\min\{a,b\}\big)\geq \prod_{i=0}^{a-1}(1-q^{i-a}).
\]
In particular, this probability tends to $1$ as $q\rightarrow\infty$.

When $a=b$, we have equality in the expression above, and
\[
\Prob\big(\dim\langle x_1,x_2,\ldots,x_a\rangle=a\big)> 1/4
\]
for any non-trivial prime power $q$.
\end{corollary}
\begin{proof} 
Write $\ell=\min\{a,b\}$. The probability we are interested in is bounded below by the probability that $x_1,x_2,\ldots, x_\ell$ are linearly independent, and is equal to this probability when $a=b$. So
\begin{align*}
\Prob\big(\dim\langle x_1,x_2,\ldots,x_b\rangle=\min\{a,b\}\big)&\geq \Prob\big(\dim\langle x_1,x_2,\ldots,x_\ell\rangle=\ell\big)\\
&=\frac{1}{q^{a\ell}}\left(\prod_{i=0}^{\ell-1}(q^a-q^{i})\right)\\
&=\prod_{i=0}^{\ell-1}(1-q^{i-a})\\
&\geq \prod_{i=0}^{a-1}(1-q^{i-a}),
\end{align*}
with equality when $a=b$. Our probability tends to $1$ as $q\rightarrow\infty$, since each factor in the last product above tends to $1$. Finally, when $a=b$ our probability is at least $1/4$ since $\prod_{i=0}^{a-1}(1-q^{i-a})=\prod_{j=1}^a(1-q^{-j})>1/4$ by Lemma~\ref{lem:qbinom}.
\end{proof}

\begin{corollary}
\label{cor:subspace_count}
Let $n$, $m$, $k$ and $t$ be fixed non-negative integers, with $t\leq \min\{n,m\}$ and $k\leq \min\{n,m\}$.
\begin{itemize}
\item[(a)] When $k+t\leq m$,
\[
\qbinom{m-k}{t} q^{tk} / \qbinom{m}{t}\rightarrow 1\text{ as }q\rightarrow\infty.
\]
\item[(b)] When $k+t\geq m$,
\[
\qbinom{k}{k+t-m} q^{(m-k)(m-t)} / \qbinom{m}{t}\rightarrow 1\text{ as }q\rightarrow\infty.
\]
\end{itemize}
\end{corollary}
\begin{proof} 
To prove~(a), we note that $\qbinom{m-k}{t} q^{tk} / \qbinom{m}{t}$ may be written as
\[
\left(\frac{\qbinom{m-k}{t}}{q^{t(m-k-t)}}\right)\left(\frac{q^{t(m-t)}}{\qbinom{m}{t}}\right).
\]
Both of the factors in the expression above tend to $1$ by~\eqref{eqn:q_big_qbinomial}, so~(a) follows.

Similarly, we may prove~(b)  by writing $\qbinom{k}{k+t-m} q^{(m-k)(m-t)} / \qbinom{m}{t}$ as
\[
\left(\frac{\qbinom{k}{k+t-m}}{q^{(k+t-m)(m-t)}}\right)\left(\frac{q^{t(m-t)}}{\qbinom{m}{t}}\right).\qedhere
\]
\end{proof}

\begin{lemma}
\label{lem:rank_k_matrices}
Let $n$, $m$ and $k$ be fixed non-negative integers such that $k\leq\min\{n,m\}$. Let $\rho(n,m,k)$ be the number of $n\times m$ matrices over~$\mathbb{F}_q$ of rank~$k$. Then
\[
\log_q \rho(n,m,k)\rightarrow (n+m-k)k\text{ as }q\rightarrow\infty.
\]
\end{lemma}
\begin{proof}
There are $\qbinom{m}{k}$ possibilities for the row space $U$ of a matrix counted by $\rho(n,m,k)$. Suppose $U$ is fixed: we wish to count the number of $n\times m$ matrices $M$ with row space $U$. Writing $x_i$ for the $i$th row of $M$, we are asking for the number of choices for $x_1,x_2,\ldots,x_n\in U$ such that $\langle x_1,x_2,\ldots,x_n\rangle =U$. Since $|U|=q^k$, the number of choices is $c\cdot q^{nk}$ where $c$ is the probability that uniformly and independently chosen vectors $x_1,x_2,\ldots,x_n\in U$ are such that $\langle x_1,x_2,\ldots,x_n\rangle =U$. Corollary~\ref{cor:span} implies that $c\rightarrow 1$ as $q\rightarrow\infty$. So
\begin{align*}
\log_q \rho(n,m,k) &= \log_q\left(\qbinom{m}{k}\right)c q^{kn}\\
&=\log_q\left(\qbinom{m}{k}\right)+\log_q c + nk\\
&\rightarrow (m-k)k+0+nk \text{ by~\eqref{eqn:log_qbinomial}}\\
&=(n+m-k)k,
\end{align*}
as required.
\end{proof}

\section{The $k$-dimensional AMMC}
\label{sec:k_ammc}

For a fixed non-negative integer $k$ such that $k\leq \min\{n,m\}$ we define the \emph{rank $k$ input additive-multiplicative matrix channel,} and write $k$-AMMC, to be the AMMC channel with input matrices $X$ restricted to having rank $k$. This section determines the asymptotic capacity of this channel when $q\rightarrow\infty$ with $n,m,t$ and $k$ fixed.

\begin{lemma}
\label{lem:symmetric}
The capacity of the $k$-AMMC is achieved when the input distribution $\bX$ is uniform over all $n\times m$ matrices $X$ of rank $k$.
\end{lemma}
Note that, throughout this paper, we use a bold letter to denote a random object (here $\bX$ is a distribution over matrices, in other words a random matrix), and the corresponding non-bold letter (here $X$) to denote a possible outcome of this random process.
\begin{proof}
The lemma is a special case of~\cite[Theorem~1]{BlackburnClaridge}. As an alternative, and more direct, proof we may show that the $k$-AMMC is symmetric in the sense of Gallager~\cite[Chapter~4]{Gallager}: we may partition the transition matrix using subsets of columns so that each part of our partition is a matrix whose rows are permutations of each other and whose columns are permutations of each other. In our case, we partition by the rank of the output matrix $Y$. (To prove that the $k$-AMMC is symmetric, we first show that $\Prob(\bY=Y|\bX=X)=\Prob(\bY=NYM|\bX=NXM)$ for any invertible $n\times n$ and $m\times m$ matrices $N$ and $M$ respectively. We then use the fact, twice, that if two $n\times m$ matrices $A$ and $B$ have the same rank, there exist an invertible $n\times n$ matrix $N$ and an invertible $m\times m$ matrix $M$ so that $NAM=B$.) Once we have shown that the $k$-AMMC is symmetric, the lemma follows by~\cite[Theorem~4.5.2]{Gallager}.
\end{proof}

\begin{lemma}
\label{lem:prob_uniform}
Let $n$, $m$, $k$ and $t$ be fixed non-negative integers, with $k\leq \min\{n,m\}$ and $t\leq\min\{m,n\}$. Consider the $k$-AMMC with uniform input distribution $\bX$ and corresponding output distribution $\bY$. For matrices $X$ and~$Y$, the probability $\Prob(\bX=X|\bY=Y)$ depends only on $\dim(\row(X)\cap\row(Y))$ and $\rk(Y)$.
\end{lemma}
\begin{proof}
For $n\times m$ matrices~$E$ and $F$, define $\Omega(E,F)$ to be the set of pairs $(W,A)$ where $W$ is an $n\times m$ matrix of rank $t$, $A$ is an invertible $n\times n$ matrix, and $F=A(E+W)$. Note that, since $\bX$ is uniform and the matrices $A$ and $W$ in the $k$-AMMC are chosen uniformly and independently, for any $n\times m$ matrices $X$ and $Y$
\begin{equation}
\label{eqn:p_and_omega}
\Prob(\bX=X\text{ and }\bY=Y)=\kappa |\Omega(X,Y)|
\end{equation}
where $\kappa$ is a constant depending only on $n$, $m$ and $t$. Indeed, in the notation of Lemma~\ref{lem:rank_k_matrices}, we have $\kappa=1/(\rho(n,m,k)\rho(n,n,n)\rho(n,m,t))$.

Let $Y$ and $Y'$ be fixed $n\times m$ matrices of rank $r$. Let $s$ be a positive integer, and let $X$ and $X'$ be two $n\times m$ matrices of rank $k$ such that
\[
\dim(\row(X)\cap\row(Y))=\dim(\row(X')\cap \row(Y'))=s.
\]
We aim to show that $|\Omega(X,Y)|=|\Omega(X',Y')|$. This equality implies that $|\Omega(X,Y)|$ depends only on $s=\dim(\row(X)\cap\row(Y))$ and $r=\rk (Y)$.

Choose a basis $y_1,y_2,\ldots,y_r$ for $\row(Y)$ such that $y_1,y_2,\ldots,y_s$ is a basis for $\row(Y)\cap \row(X)$. Extend this basis to a basis $y_1,y_2,\ldots,y_r,x_1,x_2,\ldots x_{m-r}$ for $\mathbb{F}_q^m$ in such a way that $y_1,y_2,\ldots,y_s,\allowbreak x_1,x_2,\ldots,x_{k-s}$ is a basis for $\row(X)$. Construct a similar basis $y'_1,y'_2,\ldots,y'_r,\allowbreak x'_1,x'_2,\ldots x'_{m-r}$ corresponding to $X'$ and $Y'$. So the vectors $y'_1,\ldots,y'_s$ form a basis for $\row(Y')\cap \row(X')$, the vectors $y'_1,y'_2,\ldots,y'_r$ form a basis for $\row(Y')$, and the vectors $y'_1,y'_2,\ldots,y'_s,\allowbreak x'_1,x'_2,\ldots,x'_{k-s}$ form a basis for $\row(X')$. Let $M$ be the unique (invertible) $m\times m$ matrix such that $y_iM=y'_i$ and $x_iM=x'_i$. Note that right-multiplication by $M$ maps $\row(Y)$ to $\row(Y')$ and maps $\row(X)$ to $\row(X')$. Hence $\row(YM)=\row(Y')$ and $\row(XM)=\row(X')$. Since $XM$ and $X'$ have the same row space, there exists an invertible $n\times n$ matrix $N$ such that $NXM=X'$. Since $\row(Y')=\row(YM)=\row(NYM)$, there exists an invertible $n\times n$ matrix $K$ such that $K(NYM)=Y'$. Define a map $\phi:\Omega(X,Y)\rightarrow\Omega(X',Y')$ by $\phi((W,A))=(NWM,KNAN^{-1})$. The map is a well defined bijection since the following equations are equivalent for matrices $A$ and $W$:
\begin{gather*}
A(X+W)=Y\\
KNA(X+W)M=KNYM\\
KNAN^{-1}(NXM+NWM)=Y'\\
KNAN^{-1}(X'+NWM)=Y'
\end{gather*}
So  $|\Omega(X,Y)|=|\Omega(X',Y')|$, as required.

Since $|\Omega(X,Y)|$ depends only on $\dim(\row(X)\cap\row(Y))$ and $\rk(Y)$, we see from~\eqref{eqn:p_and_omega} that the same is true for $\Prob(\bX=X\text{ and }\bY=Y)$. Moreover, writing
\[
\Prob(\bY=Y)=\sum_E \Prob(\bX=E\text{ and }\bY=Y),
\]
where the sum is over all $n\times m$ matrices $E$ of rank $k$, we see that $\Prob(\bY=Y)$ depends only on $\rk (Y)$, since the number of rank $k$ matrices $E$ with $\dim(\row(E)\cap\row(Y))$ fixed only depends on $\rk(Y)$. Hence $\Prob(\bX=X|\bY=Y)$ depends only on $\dim(\row(X)\cap\row(Y))$ and $\rk(Y)$, as required.
\end{proof}

\begin{lemma}
\label{lem:subspace_intersection}
Let $n$, $m$, $k$ and $t$ be fixed non-negative integers, with $k\leq \min\{n,m\}$ and $t\leq\min\{m,n\}$. 
Let $U$ and $V$ be sampled uniformly and independently at random from the set of all subspaces of $\mathbb{F}_q^m$ of dimension $k$ and dimension $t$ respectively.
Define $\delta=\max\{0,k+t-m\}$. Then $\dim (U\cap V)=\delta$ with probability tending to $1$ as $q\rightarrow\infty$.
\end{lemma}
\begin{proof}
Suppose first that $k+t\leq m$, and so $\delta=0$. For any fixed subspace~$U$, the proportion of subspaces $V$ that intersect $U$ trivially is
\[
\qbinom{m-k}{t} q^{tk} \, / \,\qbinom{m}{t},
\]
which tends to $1$ as $q\rightarrow\infty$, by Corollary~\ref{cor:subspace_count}. So the lemma follows in this case.

Now suppose $k+t\geq m$, and so $\delta=k+t-m$. Note that $\dim (U\cap V)=\delta$ if and only if
\[
\dim (U+ V)=\dim(U)+\dim(V)-\dim(U\cap V)=k+t-(k+t-m)=m.
\]
For any fixed subspace $U$, the proportion of subspaces $V$ that intersect $U$  in a subspace of dimension $k+t-m$ is
\[
\qbinom{k}{k+t-m}q^{(m-k)(k-(k+t-m))}\, / \,\qbinom{m}{t},
\]
which tends to $1$ as $q\rightarrow\infty$, by Corollary~\ref{cor:subspace_count}. So the lemma follows.
\end{proof}

\begin{lemma}
\label{lem:generic_dimensions}
Let $n$, $m$, $k$ and $t$ be fixed non-negative integers, with $k\leq \min\{n,m\}$ and $t\leq\min\{m,n\}$. Define $\delta=\max\{0,k+t-m\}$. Let $X$ and $W$ be $n\times m$ matrices, chosen uniformly and independently at random subject to $\rk(X)=k$ and $\rk(W)=t$. Let $Y=A(X+W)$, where $A$ is chosen uniformly at random from all $n\times n$ invertible matrices. Then with probability at least $p$, where $p\rightarrow 1$ as $q\rightarrow\infty$, we have
\begin{gather}
\rk(Y)=\min\{t+k,m,n\},\text{ and }\label{eqn:Y_rank}\\
\dim(\row(X)\cap\row(Y))=\min\{k,n-(t-\delta)\}\label{eqn:XY_intersection}.
\end{gather}
\end{lemma}
\begin{proof}
We note that the conditions~\eqref{eqn:Y_rank} and~\eqref{eqn:XY_intersection} depend only on $\row(Y)$. Since left-multiplication by an invertible matrix does not change the row space we may assume, without loss of generality, that $A$ is the identity matrix and so $Y=X+W$. So the $i$th row $y_i$ of $Y$ has the form $y_i=x_i+w_i$, where $x_i\in\mathbb{F}^m$ and $w_i\in\mathbb{F}^m$ are the $i$th rows of $X$ and $W$ respectively.

Choose subspaces $U$ and $V$ of $\mathbb{F}_q^m$ of dimensions $k$ and $t$ uniformly and independently at random.  By Lemma~\ref{lem:subspace_intersection}, $\dim (U\cap V)=\delta$ with probability $p_0$, where $p_0\rightarrow 1$ as $q\rightarrow\infty$. So to prove the lemma, it suffices to fix $U$ and $V$ such that $\dim (U\cap V)=\delta$, and then select $X$ and $W$ uniformly at random subject to $\row(X)=U$ and $\row(W)=V$: If we can show that the conditions~\eqref{eqn:Y_rank} and~\eqref{eqn:XY_intersection} hold with probability at least $p_1$ where $p_1\rightarrow 1$ as $q\rightarrow\infty$, then the lemma holds with $p=p_0p_1$.

Fix subspaces $U$ and $V$ of $\mathbb{F}_q^m$ such that $\dim(U)=k$, $\dim(V)=t$ respectively and $\dim(U\cap V)=\delta$. We generate two $n\times m$ matrices $X$ and $W$ by the following random process. We choose vectors $x_1,x_2,\ldots ,x_n\in U$ and vectors $w_1,w_2,\ldots,w_n\in V$ uniformly and independently at random. We set $X$ to be the matrix whose $i$th row is $x_i$, and we set $W$ to be the matrix whose $i$th row is $w_i$. Let $Y=X+W$, the matrix whose $i$th row is $x_i+w_i$. Now $\row(X)\subseteq U$ with equality if and only if the vectors $x_i$ span $U$: this happens with probability at least
\[
\prod_{j=1}^k (1-q^{-j}),
\]
since this expression is the probability that the first $k$ rows of $X$ span $U$. Similarly, the probability that $\row(W)=V$ is at least
\[
\prod_{j=1}^t (1-q^{-j}).
\]
Let $\bE$ be the indicator variable for the event that $\row(X)=U$  and $\row(W)=V$. Since the event that $\row(X)=U$ and the event that $\row(W)=V$ are independent, we see that
\[
\Prob(\bE=1)\geq\prod_{j=1}^k (1-q^{-j}) \cdot\prod_{j=1}^t (1-q^{-j})\rightarrow 1
\]
as $q\rightarrow\infty$.

When $\bE=1$, the matrices $X$ and $W$ are uniformly and independently chosen such that $\row(X)=U$  and $\row(W)=V$. So we are interested in the probability that the conditions~\eqref{eqn:Y_rank} and~\eqref{eqn:XY_intersection} hold conditional on $\bE=1$. However, we begin by considering the unconditional probability.

The quotient space $(U+V)/U$ has dimension $t-\delta$, since $(U+V)/U\cong V/V\cap U$. The image of the $i$th row $y_i$ of $Y$ in  $(U+V)/U$ is $y_i+U=x_i+w_i+U=w_i+U$. Since $w_1,w_2,\ldots,w_{t-\delta}$ are chosen uniformly and independently from $V$, the first $t-\delta$ rows of $Y$ span $U+V$ modulo $U$ with probability tending to $1$  as  $q \rightarrow\infty$,
by Corollary~\ref{cor:span}. When the first $t-\delta$ rows of $Y$ span $U+V$ modulo $U$, there exist field elements $\alpha_{u,j}\in\mathbb{F}_q$ with $t-\delta<u\leq n$ and $1\leq j\leq t-\delta$ such that
\[
w_u=\sum_{j=1}^{t-\delta}\alpha_{u,j}w_j\bmod U.
\]
 Consider the $n-(t-\delta)$ vectors $z_{(t-\delta)+1},z_{(t-\delta)+2},\ldots,z_n$ in $\row(Y)$ defined by
\begin{align}
\nonumber
z_u&=y_u-\sum_{j=1}^{t-\delta}\alpha_{u,j}y_j\\
&=x_u-\left(\sum_{j=1}^{t-\delta}\alpha_{u,j}x_j\right)+\left(w_u-\sum_{j=1}^{t-\delta}\alpha_{u,j}w_j\right).
\label{eqn:z_expression}
\end{align}
The last term in~\eqref{eqn:z_expression} lies in $U$, by our choice of the elements $\alpha_{i,j}$. The remaining terms lie in $U$ since $x_i\in U$. Hence $z_u\in U\cap\row(Y)$. Moreover, the vectors $z_u$ are uniformly and independently chosen from $U$, since the vectors $x_u$ for $t-\delta<u\leq n$ are only involved in the sum~\eqref{eqn:z_expression} for one value of $u$. So, by Corollary~\ref{cor:span}, with probability tending to $1$ we have a set of $n-(t-\delta)$ vectors in $\row(Y)$ which span a subspace of $U$ of dimension $\min\{k,n-(t-\delta)\}$. Thus 
$\dim(\row(X)\cap\row(Y))\geq \min\{k,n-(t-\delta)\}$ with probability tending to $1$. Clearly $\dim(\row(X)\cap\row(Y))\leq \dim(\row(X))\leq k$. Since we are assuming the first $t-\delta$ rows of $Y$ span a subspace of dimension $t-\delta$ modulo $U$, and since the row space of $Y$ is spanned by $n$ vectors, we see that $\dim(\row(Y)\cap U)\leq n-(t-\delta)$. Hence with probability tending to $1$ our random process produces matrices $X$, $W$ and $Y$ such that~\eqref{eqn:XY_intersection} holds.

We now check that our random process produces a matrix $Y$ such that~\eqref{eqn:Y_rank} holds with probability tending to $1$. Since $Y$ is an $n\times m$ matrix, $\rk(Y)\leq\min\{n,m\}$. Since $\row(Y)\subseteq \row(U)+\row(V)$ which is a space of dimension at most $t+k$, we see that $\rk(Y)\leq t+k$ and so $\rk(Y)\leq \min\{t+k,m,n\}$. The argument in the previous paragraph shows that with probability tending to $1$ the row space of $Y$ contains a linearly independent set of size $t-\delta+\min\{k,n-(t-\delta)\}=\min\{t+k-\delta,n\}$. When $t+k\leq m$ we have $\delta=0$ and so $\rk(Y)\geq \min\{t+k,n\}=\min\{t+k,n,m\}$. When $t+k>m$ then $\delta=m-(t+k)$ and so $\rk(Y)\geq \min\{m,n\}= \min\{t+k,n,m\}$. So in either case the equation~\eqref{eqn:Y_rank} holds with probability tending to $1$.

We have shown that our random process produces matrices $X$, $W$ and $Y$ such that~\eqref{eqn:Y_rank} and~\eqref{eqn:XY_intersection} hold with probabilities tending to $1$, and we have shown that $\Prob(\bE=1)=p_2$ where $p_2\rightarrow 1$ as $q\rightarrow\infty$. So~\eqref{eqn:Y_rank},~\eqref{eqn:XY_intersection} and $\bE=1$ simutaneously hold with probability $p_3$ where $p_3\rightarrow 1$ as $q\rightarrow\infty$. Hence
\begin{align*}
\Prob(\text{conditions~\eqref{eqn:Y_rank} and~\eqref{eqn:XY_intersection} hold}|\bE=1) &= \frac{\Prob(\text{\eqref{eqn:Y_rank} and~\eqref{eqn:XY_intersection} hold, and $\bE=1$})}{\Prob(\bE=1)}\\
&=p_3/p_2\rightarrow 1
\end{align*}
as $q\rightarrow\infty$. So the lemma follows.
\end{proof}

\begin{theorem}
\label{thm:k_AMMC_capacity}
Let $n$, $m$, $k$ and $t$ be fixed non-negative integers, with $k\leq \min\{n,m\}$ and $t\leq\min\{m,n\}$. Define integers $\delta$, $r$ and $s$ by $\delta=\max\{0,k+t-m\}$, $r=\min\{t+k,m,n\}$ and $s=\min\{k,n+\delta-t\}$. The capacity of the $k$-AMMC (where we take logarithms to the base $q$) tends to $(m+s-k-r)s$
as $q\rightarrow\infty$.
\end{theorem}
\begin{proof}
Let $\bX$ be a capacity achieving input distribution to the $k$-AMMC, and let $\bY$ be the corresponding output distribution. Lemma~\ref{lem:symmetric} shows that we may take $\bX$ to be the uniform distribution on the $n\times m$ matrices of rank~$k$. We need to calculate $I(\bX;\bY)=H(\bX)-H(\bX|\bY)$ in this case.

We see that $H(\bX)\rightarrow (m-k)k+nk$, by Lemma~\ref{lem:rank_k_matrices}, since $\bX$ is uniform. We now provide a value for the limit of $H(\bX|\bY)$.

By Lemma~\ref{lem:generic_dimensions},
\begin{align*}
H(\bX|\bY)&=\sum_{r=0}^{\min\{n,m\}}\sum_{s=0}^{\min\{r,k\}} \big(\Prob(\rk(Y)=r,\dim(\row(X)\cap\row(Y))=s))\cdot\\
&\quad\quad\quad\quad\quad\quad\quad H(\bX|\bY,\rk(\bY)=r,\dim(\row(\bY)\cap\row(\bX))=s)\big)\\
&\geq p\, H\big(\bX|\bY,\rk(\bY)=\min\{t+k,m,n\},\\
&\quad\quad\quad\quad\quad\quad\dim(\row(\bY)\cap\row(\bX))=\min\{k,n-(t-\delta)\}\big),
\end{align*}
where $p\rightarrow 1$ as $q\rightarrow\infty$. Similarly,
\begin{align*}
H(\bX|\bY)&\leq p\, H\big(\bX|\bY,\rk(\bY)=\min\{t+k,m,n\},\\
&\quad\quad\quad\quad\quad\quad\dim(\row(\bY)\cap\row(\bX))=\min\{k,n-(t-\delta)\}\big)\\
&\quad +(1-p)mn.
\end{align*}
Since $(1-p)mn\rightarrow 0$ as $q\rightarrow\infty$,
\begin{align}
H(\bX|\bY)&\rightarrow H\big(\bX|\bY,\rk(\bY)=\min\{t+k,m,n\},\nonumber\\
&\quad\quad\quad\quad\quad\quad\dim(\row(\bY)\cap\row(\bX))=\min\{k,n-(t-\delta)\}\big).
\label{eqn:conditional_entropy}
\end{align}

By Lemma~\ref{lem:prob_uniform}, the distribution of $\bX$ given $\rk(Y)=r$ and $\dim(\row(Y)\cap\row(\bX))=s$ is uniform, for any fixed $r$ and $s$. In this situation, the number of possibilities for the row space of $\bX$ is
\[
\qbinom{r}{s}\qbinom{m-r}{k-s}q^{(k-s)(r-s)}.
\]
Lemma~\ref{lem:qbinom} shows that the logarithm to the base $q$ of this expression tends to $(r-s)s+(m-k+s-r)(k-s)+(k-s)(r-s)=(r-s)s+(m-k)(k-s)$. By Corollary~\ref{cor:span}, the logarithm to the base $q$ of the number of possibilities for $\bX$ tends to $(r-s)s+(m-k)(k-s)+nk$, and so the entropy of $\bX$ given $\rk(Y)=r$ and $\dim(\row(Y)\cap\row(\bX))=s$ tends to $(r-s)s+(m-k)(k-s)+nk$. Hence, by~\eqref{eqn:conditional_entropy},
$H(\bX|\bY)\rightarrow (r-s)s+(m-k)(k-s)+nk$ where $r=\min\{t+k,m,n\}$ and $s=\min\{k,n-(t-\delta)\}$. Thus
\begin{align*}
I(\bX;\bY)&\rightarrow \left((m-k)k+nk\right)-\left((r-s)s+(m-k)(k-s)+nk\right)\\
&=(m-k)s-(r-s)s
\end{align*}
and the theorem follows.
\end{proof}

\section{The AMMC}
\label{sec:AMMC}

We are now in a position to establish the capacity of the AMMC when $q\rightarrow\infty$.

\begin{theorem}
\label{thm:AMMC_capacity}
Let $n$,$m$ and $t$ be fixed non-negative integers such that $t\leq \min\{m,n\}$. When $m+t\geq 2n$, the capacity of the AMMC tends to
$(m-n)(n-t)$ as $q\rightarrow\infty$. Otherwise, the capacity of thfe AMMC tends to $\lceil (m-t)/2\rceil\cdot \lfloor (m-t)/2\rfloor$ as $q$ tends to $\infty$. (We take logarithms to the base $q$ when computing capacity.)
\end{theorem}
\begin{proof}
We need to find the maximal mutual information $I(\bX;\bY)$ over all input distributions $\bX$ to the AMMC, where $\bY$ is the corresponding output distribution.

Let $c$ be the maximum capacity of the $k$-AMMC, as $k$ ranges over the set $\{0,1,\ldots,\min\{n,m\}\}$. If $\bX$ is a capacity achieving distribution for the $k$-AMMC, then we may use this input distribution to the AMMC. So the capacity of the $k$-AMMC is a lower bound for the capacity of the AMMC. In particular, $c$ is a lower bound for the capacity of the AMMC.

We aim to derive an upper bound for the capacity of the AMMC, also involving $c$. Let $\bX$ be an input distribution to the AMMC and $\bY$ the corresponding output distribution. The data processing inequality implies that
$I(\bX;\bY,\rk(\bX))\geq I(\bX;\bY)$,
and so using the chain rule we find
\begin{align}
\nonumber
I(\bX;\bY)&\leq I(\bX;\bY,\rk(\bX))\\
\nonumber
&=I(\bX;\rk(\bX))+I(\bX;\bY|\rk(\bX))\\
\nonumber
&=H(\rk(\bX))+I(\bX;\bY|\rk(\bX))\\
&\leq\log_q(\min\{n,m\}+1)+I(\bX;\bY|\rk(\bX)),
\label{eqn:rank_term}
\end{align}
since $\rk(\bX)$ is a random variable with support at most $\min\{n,m\}+1$.

But $c$ is an upper bound for $I(\bX;\bY|\rk(\bX))$, since
\begin{align}
I(\bX,\bY|\rk(\bX))&= \sum_{k=0}^{\min\{m,n\}}\Prob(\rk(\bX)=k)\cdot I(\bX,\bY|\rk(\bX)=k)\nonumber\\
&\leq \sum_{k=0}^{\min\{m,n\}}\Prob(\rk(\bX)=k)\cdot c=c.
\label{eqn:max_information}
\end{align}

The inequalities~\eqref{eqn:rank_term} and~\eqref{eqn:max_information} combine to show that the capacity of the AMMC is bounded above by $c+\log_q\min\{n,m\}$. Since $\log_q\min\{n,m\}\rightarrow 0$ as $q\rightarrow\infty$, and since we have a lower bound of $c$ on the capacity, we see that the capacity of the AMMC tends to $c$ as $q\rightarrow\infty$. So to prove the theorem, it suffices to determine the value of $c$.

Define $\delta$, $r$ and $s$ as in Theorem~\ref{thm:k_AMMC_capacity}. Suppose first that $k$ is such that $k+t>m$ and so $\delta=(k+t)-m$. We claim that the capacity of the $k$-AMMC tends to $0$ in this situation. To prove this claim, assume first that $m\leq n$. Then $r=\min\{k+t,n,m\}=m$ and $s=\min\{k,n+k+t-m-t\}=k$. But then $m+s-k-r=0$ and so the capacity tends to $0$ by Theorem~\ref{thm:k_AMMC_capacity}. Now assume that $n\leq m$. Here $r=n$ and $s=k+n-m$. So again $m+s-k-r=m-n=0$ and the capacity again tends to $0$, establishing our claim.

Now assume that $k+t\leq m$, and so $\delta=0$. We have $r=\min\{t+k,n\}=t+\min\{k,n-t\}$ and $s=\min\{k,n-t\}$. So $r-s=t$, and Theorem~\ref{thm:k_AMMC_capacity} implies that the capacity of the $k$-AMMC tends to $(m-t-k)s$. When $k\geq n-t$ this capacity is $(m-t-k)(n-t)$, which is a decreasing function of $k$. As we are maximising the capacity, we may assume without loss of generality that $k\leq n-t$ and the capacity is $(m-t-k)k$. Note that, since we are in the situation when $k+t\leq m$ and $k\leq n-t$, the integer $k$ lies in the interval $[0,\min\{n-t,m-t\}]$.

The function $f(x)=(m-t-x)x$ has a unique maximum at $x=(m-t)/2$. Now $x\geq 0$ as $t\leq \min\{m,n\}\leq m$, and clearly $x\leq m-t$. Moreover $x< n-t$ if and only if $m+t< 2n$. We have two cases to consider. We first suppose that $m+t\geq 2n$. Then $x\geq n-t\geq \min\{n-t,m-t\}$ and so $x$ lies above the interval $[0,\min\{n-t,m\})$. Hence in this case $f$ is an increasing function on $[0,\min\{n-t,m-t\}]$ and so the capacity is maximal when $k=\min\{n-t,m-t\}$. But note that $m\geq 2n-t\geq n$ (as  $m+t\geq 2n$ and $t\leq\min\{n,m\}$) and so the capacity is maximised  when $k=\min\{n-t,m-t\}=n-t$ at the value $(m-n)(n-t)$. Hence the AMMC capacity is $(m-n)(n-t)$ when $m+t\geq 2n$, and the first case of the theorem follows. Secondly, we suppose that $m+t< 2n$. In this case $x$ lies in $[0,\min\{n-t,m\}]$ and so the capacity is maximised when $k=\lfloor x\rfloor=\lfloor (m-t)/2\rfloor$ at the value $\lceil (m-t)/2\rceil\cdot \lfloor (m-t)/2\rfloor$. This establishes the second case of the theorem, as required.
\end{proof}

\section{When the Field size is Fixed}
\label{sec:fixed_field}

Throughout this section, we assume that the cardinality $q$ of the underlying finite field is fixed. Let $\nu,\mu,\tau\in (0,1)$ be fixed real numbers such that $\tau\leq\min\{\nu,\mu\}$. For a positive integer $\ell$ (which we think of as a size parameter), set $n=\lfloor \nu \ell\rfloor$, $m=\lfloor\mu \ell\rfloor$, $t=\lfloor \tau\ell\rfloor$. We aim (Theorem~\ref{thm:fixed_q_AMMC} below) to determine a good approximation for the capacity of the AMMC as $\ell\rightarrow\infty$ with $\nu$, $\mu$ and $\tau$ fixed. Indeed, we will show that the capacity of the AMMC is $(\mu-\nu)(\nu-\tau)\ell^2+O(\ell\log_q \ell)$ when $\mu+\tau\geq 2\nu$, and is $\frac{1}{4}(\mu-\tau)\ell^2+O(\ell\log_q\ell)$  otherwise. This generalises and strengthens~\cite[Proposition~4]{SilvaKoetter}, which shows, in our notation, that the capacity is $(\mu-\nu)(\nu-\tau)\ell^2+o(\ell^2)$ when $\mu\geq 2\nu$. (Note that the definition of $\tau$ in~\cite{SilvaKoetter} and our paper are different. Note also that the authors in~\cite{SilvaKoetter} phrase their result in terms of \emph{normalised capacity}, which is the result of dividing the capacity by $nm$. So their formulae differ from  ours by a factor of $\nu\mu\ell^2$.)

We will mimic the proofs in Sections~\ref{sec:k_ammc} and~\ref{sec:AMMC}, and only include details at points where the proofs differ. We note that Lemmas~\ref{lem:symmetric} and~\ref{lem:prob_uniform} hold in our situation without change.  We will prove analogues of Lemmas~\ref{lem:subspace_intersection} and~\ref{lem:generic_dimensions}  and Theorem~\ref{thm:k_AMMC_capacity} in order to provide good bounds on the $k$-AMMC in our asymptotic regime. Just as in Section~\ref{sec:k_ammc}, the most difficult part of the proof is the analogue of Lemma~\ref{lem:generic_dimensions}. In Section~\ref{sec:k_ammc} we showed that the rank of the matrix $Y$ and the dimension of the intersection $\row(X)\cap\row(Y)$ can both be predicted with high probability in the $k$-AMMC when $q\rightarrow\infty$ with $n,m$ and $t$ fixed. In the asymptotic regime we are considering in this section, this is no longer true. But we are able to exploit the fact that the approximate values of $\rk(Y)$ and $\dim(\row(X)\cap\row(Y))$ can be found with high probability. 

The following lemma is an analogue of Lemma~\ref{lem:subspace_intersection}:

\begin{lemma}
\label{lem:qfixed_subspace_intersection}
Define $\mu$, $\tau$, $\ell$, $m$ and $t$ as above. Let $k$ be a non-negative integer such that $k\leq m$. 
Let $U$ and $V$ be uniformly and independently sampled from the set of all subspaces of $\mathbb{F}_q^m$ of dimensions $k$ and $t$ respectively. Define $\delta=\max\{0,(k+t)-m\}$. Then whenever $\ell>\sqrt{t}q^2$,
\[
\dim(U\cap V)\in [\,\delta,\delta+3\sqrt{\log_q\ell}\,],
\]
with probability at least $1-1/\ell^7$.
\end{lemma}
\begin{proof}
We see that $\dim (U+V)=\dim U+\dim V-\dim(U\cap V)$. So $m\geq \dim (U+V)=k+t-\dim(U\cap V)$. Hence $\dim (U\cap V)\geq \delta$ in all cases. So it remains to show that $\dim (U\cap V)\leq \delta+3\sqrt{\log_q\ell}$ with high probability.

For any fixed subspace $V$, the number of choices for $U$ with $\dim (U\cap V)=\delta+b$, say, is
\[
\qbinom{t}{\delta+b}\qbinom{m-t}{k-\delta-b}q^{(k-\delta-b)(t-\delta-b)},
\]
for $0\leq b\leq \min\{k,t\}-\delta$, and is zero otherwise. By Lemma~\ref{lem:qbinom}, the logarithm to the base $q$ of this expression is bounded above by
\[
(t-\delta-b)(\delta+b)+\log_q4+(m-t-k+\delta+b)(k-\delta-b)+\log_q4+(k-\delta-b)(t-\delta-b).
\]
Since the number of subspaces of $\mathbb{F}_q^m$ of dimension $k$ is at least $q^{(m-k)k}$, by Lemma~\ref{lem:qbinom}, we see that the logarithm to the base $q$ of the probability that $\dim (U\cap V)=\delta+b$ is at most
\[
\begin{split}
(t-\delta-b)(\delta+b)&+\log_q4+(m-t-k+\delta+b)(k-\delta-b)
+\log_q4\\
&+(k-\delta-b)(t-\delta-b)-(m-k)k,\end{split}
\]
which is equal to
\begin{equation}
\label{eqn:prob_bound}
(k+t-m-\delta-b)(\delta+b)
+2\log_q4.
\end{equation}

We claim that the expression~\eqref{eqn:prob_bound} is bounded above by $-b^2+4$. To see this, first note that when $k+t\leq m$, and so $\delta=0$, the expression~\eqref{eqn:prob_bound} can be written as $-b^2+b((k+t)-m)+2\log_q 4$, which is bounded above by
$-b^2+4$ since $\log_q4\leq 2$. When $k+t>m$, and so $\delta=k+t-m$, the same bound holds, since the expression~\eqref{eqn:prob_bound} can be written as
$-b^2+(m-(k+t))b+2\log_q 4$ which is bounded above by $-b^2+4$ also. So the claim is established.

Hence the probability that $\dim(U\cap V)>\delta+3\sqrt{\log_q \ell}$ is at most
\[
\sum_{b=3\sqrt{\log_q \ell}}^{\min\{k,t\}-\delta} q^{4-b^2}\leq tq^4q^{-9\log_q\ell}=(tq^4/\ell^2)/\ell^7.
\]
Since $\ell>\sqrt{t}q^2$, we see that $tq^4/\ell^2<1$ and so the lemma follows.
\end{proof}

Next, we use Lemma~\ref{lem:qfixed_subspace_intersection} to prove an analogue of Lemma~\ref{lem:generic_dimensions}:

\begin{lemma}
\label{lem:qfixed_generic_dimensions}
Define $\nu$, $\mu$, $\tau$, $\ell$, $n$, $m$ and $t$ as above. Let $k$ be an integer such that $0\leq k\leq \min\{n,m\}$, and define $\delta=\max\{0,k+t-m\}$ as before. Define $f=13\log_q \ell$. Let $X$ and $W$ be $n\times m$ matrices, chosen uniformly and independently at random subject to $\rk(X)=k$ and $\rk(W)=t$. Let $Y=A(X+W)$, where $A$ is chosen uniformly at random from all $n\times n$ invertible matrices. Whenever $\ell>16+\sqrt{1+t+n}+q^9$, we have
\begin{gather}
\rk(Y)\in \min\{t+k,n,m\}+[-f,0],\text{ and }\label{eqn:qfixed_Y_rank}\\
\dim(\row(X)\cap\row(Y))\in \min\{k,n-(t-\delta)\}+[-f,f]\label{eqn:qfixed_XY_intersection}.
\end{gather}
with probability $p$, where $p>1-1/\ell^3$.
\end{lemma}
\begin{proof}
As in the proof of Lemma~\ref{lem:generic_dimensions}, we may assume without loss of generality that $A$ is the identity matrix, and so $Y=X+W$.

Just as in the proof of Lemma~\ref{lem:generic_dimensions}, we choose the matrices $X$ and $W$ using the following random process. First we choose subspaces $U$ and $V$ of $\mathbb {F}_q^m$ of dimensions $k$ and $t$ respectively, uniformly and independently at random. We then choose the rows $x_i$ of $X$ independently and uniformly from $U$, and the rows $w_i$ of $W$ independently and uniformly from $V$. Let $\bE$ be the indicator variable for the event that $\row(X)=U$ and $\row(W)=V$. When $\bE=1$ the matrices $X$ and $W$ are uniform such that $\row(X)=U$ and $\row(W)=V$. So the probability $p$ we are interested in is the probability that~\eqref{eqn:qfixed_Y_rank} and~\eqref{eqn:qfixed_XY_intersection} hold, given that $\bE=1$.

Following the proof of Lemma~\ref{lem:generic_dimensions}, we see that
\[
\Prob(\bE=1)\geq \prod_{j=1}^{k}(1-q^{-j})\prod_{j=1}^{t}(1-q^{-j}).
\]
Unlike in Lemma~\ref{lem:generic_dimensions}, the right hand side of this expression does not tend to~1. However, by Lemma~\ref{lem:qbinom},
 $\Prob(\bE=1)> (1/4)^2=1/16$. 

Let $\pi$ be the `bad' probability that our random process produces matrices $X$ and $Y$ such that~\eqref{eqn:qfixed_Y_rank} or~\eqref{eqn:qfixed_XY_intersection} fail.  Then the probability $1-p$ that~\eqref{eqn:qfixed_Y_rank} or~\eqref{eqn:qfixed_XY_intersection} fail given $\bE=1$ is at most $\pi/\Prob(\bE=1)\leq 16\pi$. We aim to show that $\pi\leq 1/\ell^4$ for the values of $\ell$ we consider. This is sufficient to prove the lemma, since $16\pi\leq 16/\ell^4\leq 1/\ell^3$ as we are assuming that $\ell\geq 16$.

With probability at least $1-1/\ell^7$, $\dim(U\cap V)=z$, where $z\in [\delta,\delta+3\sqrt{\log_q \ell}]$, by Lemma~\ref{lem:qfixed_subspace_intersection}. Assume that $z\in [\delta,\delta+3\sqrt{\log_q \ell}]$, so
\[
t-\delta-3\sqrt{\log_q\ell}\leq \dim(V\bmod U)\leq t-\delta.
\]
Define $r_i$ to be the dimension of the space spanned by the first $i$ rows of $Y$, so $r_i=\dim \langle y_1,y_2,\ldots ,y_i\rangle$. Now $y_i=x_i+w_i=w_i\bmod U$, and so the vectors $y_i\bmod U$ are uniformly and independently chosen elements from $V\bmod U$. The probability that the first $t-\delta-3\sqrt{\log_q\ell}-6\log_q\ell$ of the vectors $y_i$ are linearly independent modulo $U$ is
\begin{align*}
\prod_{i=1}^{t-\delta-3\sqrt{\log_q\ell}-6\log_q\ell}(q^{\dim(U\cap V)}-q^i)/q^{\dim(U\cap V)}&\geq\left(1-q^{-6\log_q\ell}\right)^{t}\\
&\geq 1-t/\ell^6.
\end{align*}
So $r_i=i$ for $i\leq t-\delta-3\sqrt{\log_q\ell}-6\log_q\ell$ with probability at least $1-t/\ell^6$. Suppose that this high probability event occurs: $r_i=i$ for $i\leq t-\delta-3\sqrt{\log_q\ell}-6\log_q\ell$. Consider the next $\min\{k,n-(t-\delta)\}-6\log_q\ell$ rows of $Y$. For at most $3\sqrt{\log_q\ell}+6\log_q\ell$ of these vectors $y_i$, it is possible that $y_i\bmod U$ is not in the span of $y_1,y_2\ldots ,y_{i-1}\bmod U$ and so $r_i=r_{i-1}+1$ with probability~$1$. Otherwise, the argument in the proof of Lemma~\ref{lem:generic_dimensions} shows that there is a linear combination of $y_i$ and previous rows that produces a uniformly chosen element of $U$; with probability at least $1-q^{-6\log_q\ell}=1-1/\ell^6$ this vector does not lie in the span of previous rows. Hence the event $E_0$ that
\begin{gather*}
\text{$r_i=i$ for }i\leq t-\delta-3\sqrt{\log_q\ell}-6\log_q\ell+\min\{k,n-(t-\delta)\}-6\log_q\ell\\
\text{ and }\dim(U\cap V)\in [\delta,\delta+3\sqrt{\log_q \ell}]
\end{gather*}
occurs with probability at least
\[
(1-1/\ell^7)(1-t/\ell^6)(1-1/\ell^6)^{\min\{k,n-(t-\delta)\}-6\log_q\ell}.
\]
We see that this probability is bounded below by $(1-1/\ell^7)(1-t/\ell^6)(1-n/\ell^6)$, which is at least $1-1/\ell^7-t/\ell^6-n/\ell^6$. Now $\ell\geq \sqrt{1+t+n}$, and so
\begin{align*}
1-1/\ell^7-t/\ell^6-n/\ell^6&\geq 1-(1+t+n)/\ell^6\\
&\geq 1-\ell^2/\ell^6.
\end{align*}
So $E_0$ occurs with probability at least $1-1/\ell^4$.

When the event $E_0$ occurs, since $\rk(Y)\geq r_i$ for any $i$, we see that
\begin{align*}
\rk(Y)&\geq t-\delta-3\sqrt{\log_q\ell}-6\log_q\ell+\min\{k,n-(t-\delta)\}-6\log_q\ell\\
&\geq \min\{k+t-\delta,n\}-f
\end{align*}
since $\ell\geq q^9$ (which implies that $3\sqrt{\log_q\ell}\leq\log_q\ell$). Note that the definition of $\delta$ shows that $\min\{k+t-\delta,n\}=\min\{k+t,n,m\}$. Since $\row(Y)\subseteq U+V$, we see that $\rk(Y)\leq k+t$. Moreover, the rank of an $n\times m$ matrix can be at most $\min\{n,m\}$. Hence~\eqref{eqn:qfixed_Y_rank} holds when $E_0$ occurs. 

We have $\row(X)+\row(Y)=\row(X)+\row(W)$. Hence, considering the dimensions of these subspaces, we find
\[
\rk(X)+\rk(Y)-\dim(\row(X)\cap\row(Y))=\rk(X)+\rk(W)-\dim(U\cap V),
\]
and so
\begin{align*}
\dim(\row(X)\cap\row(Y))&=\rk(Y)-\rk(W)+\dim(U\cap V)\\
&=\rk(Y)-t+\dim(U\cap V).
\end{align*}
When $E_0$ holds,
\[
\dim(U\cap V)\in [\delta,\delta+3\sqrt{\log_q \ell}]\subseteq [\delta,\delta+f]
\]
and $\rk(Y)\in \min\{k+t-\delta,n\}+[-f,0]$. Hence
\[
\dim(\row(X)\cap\row(Y))\in \min\{k,n-(t-\delta)\}+[-f,f],
\]
and so $\eqref{eqn:qfixed_XY_intersection}$ holds. Hence the lemma follows.
\end{proof}

We are now in a position where we can find a good approximation for the capacity of the $k$-AMMC. The following theorem is an analogue of Theorem~\ref{thm:k_AMMC_capacity}:

\begin{theorem}
\label{thm:qfixed_k_AMMC_capacity}
Let $q$ be a fixed prime power.
Let $\nu,\mu,\tau\in (0,1)$ be fixed real numbers with the property that $\tau\leq\min\{\nu,\mu\}$. For a positive integer $\ell$, set $n=\lfloor \nu \ell\rfloor$, $m=\lfloor\mu \ell\rfloor$ and $t=\lfloor \tau\ell\rfloor$. Let $k$ be an integer such that $0\leq k\leq \min\{n,m\}$. Define integers $\delta$, $r$ and $s$ by $\delta=\max\{0,(k+t)-m\}$, $r=\min\{t+k,m,n\}$ and $s=\min\{k,n+\delta-t\}$. Define $f=13\log_q\ell$. The capacity of the $k$-AMMC (taking logarithms to the base $q$) lies between
\[
(m+s-k-r)s-\big(2(m+n)f+(m+n)^2/\ell^3+\log_q 4\big)
\]
and
\[
(m+s-k-r)s+\big((m+n)f+\log_q4+4(m+n)^2/\ell^3\big)
\]
whenever $\ell\geq 16+\sqrt{1+t+n}+q^9+14^2/(\nu+\mu)^2$. 
\end{theorem}
Note that $m$ and $n$ are linear in $\ell$, whereas $f$ is logarithmic. So the theorem says that the capacity of the $k$-AMMC is approximately $(m+s-k-r)s$. 
\begin{proof}
As before, we take $\bX$ to be the uniform input distribution for the $k$-AMMC, and $\bY$ to be the corresponding output distribution. This is a capacity achieving distribution, by Lemma~\ref{lem:symmetric}, and so we can estimate the capacity by estimating $H(\bX)-H(\bX|\bY)$. The number of $n\times m$ matrices of rank $k$ is at most $\qbinom{m}{k} q^{nk}$, and is at least $\qbinom{m}{k}q^{nk}/4$, since the first $k$ rows of a matrix whose rows lie in a subspace of dimension $k$ span the subspace with probability at least 1/4, by Corollary~\ref{cor:span}. Since $\bX$ is uniform, Lemma~\ref{lem:qbinom} implies that
\begin{equation}
\label{eqn:HX}
(m-k)k+nk-\log_q4\leq H(\bX)\leq (m-k)k+nk+\log_q4,
\end{equation}
just as in the proof of Theorem~\ref{thm:k_AMMC_capacity}.

To estimate $H(\bX|\bY)$, define $I$ to be the set of pairs $(r',s')$ of non-negative integers such that:
\begin{gather*}
r'\in \min\{t+k,n,m\}+[-f,0],\text{ and}\\
s'\in \big(\min\{k,n-(t-\delta)\}+[-f,f]]\big)\cap [0,\min\{r,k\}]
\end{gather*}
Define 
\begin{align*}
\underline{h}=\min_{(r',s')\in I}\big\{\Prob(\rk(Y)&=r',\dim(\row(X)\cap\row(Y))=s')\cdot\\
&H(\bX|\bY,\rk(\bY)=r',\dim(\row(\bY)\cap\row(\bX))=s')\big\}
\end{align*}
and
\begin{align*}
\overline{h}=\max_{(r',s')\in I}\big\{\Prob(\rk(Y)&=r',\dim(\row(X)\cap\row(Y))=s')\cdot\\
&H(\bX|\bY,\rk(\bY)=r',\dim(\row(\bY)\cap\row(\bX))=s')\big\}.
\end{align*}
Lemma~\ref{lem:qfixed_generic_dimensions} and the method in the proof of Theorem~\ref{thm:k_AMMC_capacity} shows that
\begin{equation}
\label{eqn:HXY}
p\underline{h}\leq H(\bX|\bY)\leq p\overline{h}+(1-p)mn,
\end{equation}
where $1\geq p\geq 1-1/\ell^3$.
We note that, in our asymptotic regime, $mn/\ell^3\rightarrow 0$. And so to estimate $H(\bX|\bY)$ it suffices to provide a lower bound for $\underline{h}$ and an upper bound for $\overline{h}$.

When $Y$ is fixed such that $\rk(Y)=r'$, the number of choices for $\row(X)$ such that $\dim(\row(X)\cap\row(Y))=s'$ is
\[
\qbinom{r'}{s'} \qbinom{m-r'}{k-s'}q^{(k-s')(r'-s')}.
\]
The number of matrices $X$ with a given row space $U$ of rank $k$ is at most $q^{kn}$ (as every row of $X$ lies in $U$). At least $1/4$ of the matrices $X$ whose rows lie in $U$ have $\row(X)=U$, since $k$ randomly chosen vectors from $U$ span $U$ with probability at least $1/4$, by Corollary~\ref{cor:span}. So the number of matrices $X$ with a given row space $U$ of rank $k$ is at least $q^{kn}/4$. By Lemma~\ref{lem:qbinom}, we see that the number of possibilities for $X$ is at most $q$ to the power
\begin{align*}
(r'-s')s'+&\log_q 4 +(m-r'-k+s')(k-s')\\
&\quad +\log_q 4+ (k-s')(r'-s')+kn\\
&= (r'-s')s' +(m-k)(k-s')+kn+2\log_q 4
\end{align*}
and, similarly, at least $q$ to the power
\[
(r'-s')s' +(m-k)(k-s')+kn-\log_q4.
\]
These expressions are upper and lower bounds for the conditional entropy $H(\bX|\bY,\rk(\bY)=r',\dim(\row(\bY)\cap\row(\bX))=s')$. 

When $(r',s')\in I$ we see that $r-f\leq r'\leq r$ and $s-f\leq s'\leq s+f$, where $r$ and $s$ are defined in the statement of the theorem. Hence
\begin{align*}
\overline{h}&\leq (r-s+f)(s+f)+(m-k)(k-s+f)+kn+2\log_q 4\\
&=(r-s)s+(m-k)(k-s)+kn+2\log_q 4+f^2+rf+(m-k)f.
\end{align*}
Now, $rf+(m-k)f\leq (m+n)f$, since $r\leq t+k\leq n+k$. Since $f\geq 16$, we see that $2\log_q4<f$ and so $2\log_q 4+f^2\leq f(f+1)$. Hence
$\overline{h}\leq (r-s)s+(m-k)(k-s)+kn+(f+1)f+(m+n)f$.
Also
\[
f+1\leq 14\log_q\ell\leq 14\sqrt{\ell}-2=14\ell/\sqrt{\ell}-2\leq (\nu+\mu)\ell-2\leq m+n,
\]
since $\ell\geq 14^2/(\nu+\mu)^2$. So
\[
\overline{h}\leq (r-s)s+(m-k)(k-s)+kn+2(m+n)f.
\]
Using this bound together with the expressions~\eqref{eqn:HX} and~\eqref{eqn:HXY} we see that
\begin{align*}
I(\bX;\bY)&=H(\bX)-H(\bX|\bY)\\
&\geq (m-k)k+nk-\log_q 4-p\overline{h}-(1-p)mn\\
&\geq (m-k)k+nk-\log_q 4-\overline{h}-mn/\ell^3\\
&\geq (m-k)s+(s-r)s-\log_q4-2(m+n)f-(m+n)^2/\ell^3.
\end{align*}
Hence the lower bound of the theorem follows.

Turning to $\underline{h}$, we see that
\begin{align*}
\underline{h}&\geq (r-s-2f)(s-f) +(m-k)(k-s-f)+kn-\log_q4\\
&=(r-s)s +(m-k)(k-s)+kn\\
&\quad-\log_q4-2fs+2f^2-(r-s)f-(m-k)f\\
&\geq (r-s)s +(m-k)(k-s)+kn-(r+s)f-(m-k)f\\
&\quad\quad\text{(as $2f^2\geq \log_q 4$)}\\
&\geq (r-s)s +(m-k)(k-s)+kn-(m+n)f
\end{align*}
since $s\leq k$ and $r\leq n$.
Using this bound and the expressions~\eqref{eqn:HX} and~\eqref{eqn:HXY}, we see that
\begin{align*}
I(\bX;\bY)&=H(\bX)-H(\bX|\bY)\\
&\leq  (m-k)k+nk+\log_q 4-p\underline{h}\\
&\leq (m-k)k+nk+\log_q 4\\
&\quad \quad-(1-1/\ell^3)\big(
(r-s)s+(m-k)(k-s)+kn-(m+n)f\big)\\
&\leq (m-k)s+(s-r)s+(m+n)f+\log_q4+4(m+n)^2/\ell^3.
\end{align*}
Hence we have established the upper bound on the capacity, and the theorem is proved.
\end{proof}

\begin{theorem}
\label{thm:fixed_q_AMMC}
Let $q$ be a fixed prime power.
Let $\nu,\mu,\tau\in (0,1)$ be fixed real numbers with the property that $\tau\leq\min\{\nu,\mu\}$. For a positive integer $\ell$, set $n=\lfloor \nu \ell\rfloor$, $m=\lfloor\mu \ell\rfloor$ and $t=\lfloor \tau\ell\rfloor$. The capacity of the AMMC (taking logarithms to the base $q$) is $(\mu-\nu)(\nu-\tau)\ell^2+O(\ell\log_q\ell)$ when $\mu+\tau\geq 2\nu$, and is $\frac{1}{4}(\mu-\tau)\ell^2+O(\ell\log_q\ell)$ otherwise.
\end{theorem}
\begin{proof}
Let $c$ be the maximum capacity of the $k$-AMMC, where $0\leq k\leq \min\{m,n\}$. The argument in the proof of Theorem~\ref{thm:AMMC_capacity} shows that $c$ and the capacity of the AMMC differ by at most $\log_q(\min\{n,m\}+1)=O(\log_q\ell)$ (where the implicit constant depends on $\mu$ and $\nu$, but not on $k$). We use Theorem~\ref{thm:qfixed_k_AMMC_capacity} to estimate the capacity of the $k$-AMMC. In that theorem, note that $(m+n)f=O(\ell\log_q\ell)$ and $(m+n)^2/\ell^3=O(1/\ell)$, where the implicit constants do not depend on $k$. So Theorem~\ref{thm:qfixed_k_AMMC_capacity} shows that the capacity of the $k$-AMMC is $(m+s-k-r)s+O(\ell\log_q\ell)$, where $r$ and $s$ are defined in the theorem. We may now maximise the value of $(m+s-k-r)s$ to find~$c$: the argument of Theorem~\ref{thm:AMMC_capacity} shows $c=(m-n)(n-t)+O(\ell\log_q\ell)$ when $m+t\geq 2n$ and $c=\lceil (m-t)/2\rceil\cdot \lfloor (m-t)/2\rfloor$ otherwise. Since $(m-n)(n-t)=(\mu-\nu)(\nu-\tau)\ell^2+O(\ell)$, and $\lceil (m-t)/2\rceil\cdot \lfloor (m-t)/2\rfloor=(\mu-\tau)^2/4+O(\ell)$, the theorem follows.
\end{proof}

\section{A Coding Scheme}
\label{sec:coding_scheme}

A straightforward modification of the coding scheme in~\cite{SilvaKoetter} works in our more general setting. Consider the situation when $q$ is large. We think of our data as an $x\times y$ matrix $U$, where $x=n-t$, $y=m-n$ when $m+t\geq 2n$, and where $x=\lfloor (m-t)/2\rfloor$ and $y=\lceil (m-t)/2\rceil$  otherwise. Note that in both cases, $x+y=m-t$. Writing $0_{a\times b}$ for the $a\times b$ zero matrix, and writing $I_a$ for the $a\times a$ identity matrix, we input $X$ into the channel, where $X$ is the $n\times m$ matrix defined by
\[
X=\left(\begin{array}{c|c|c}
0_{(n-x)\times t}&0_{(n-x)\times x}&0_{(n-x)\times y}\\\hline
0_{x \times t}&I_x&U
\end{array}\right).
\]
So when $m+t\geq 2n$ we have
\[
X=\left(\begin{array}{c|c|c}
0_{t\times t}&0_{t\times n-t}&0_{t\times m-n}\\\hline
0_{(n-t)\times t}&I_{n-t}&U
\end{array}\right),
\]
and if not, setting $d=\lfloor (m-t)/2\rfloor$ and $d'=\lceil (m-t)/2\rceil$, we have
\[
X=
\left(\begin{array}{c|c|c}
0_{(n-d)\times t}&0_{(n-d)\times d}&0_{(n-d)\times d'}\\\hline
0_{d\times t}&I_{d}&U
\end{array}\right).
\]
On receiving the corresponding matrix $Y=A(X+W)$, we row reduce to produce a matrix~$Z$. Lemma~\ref{lem:generic_dimensions} shows that with high probability we expect $Y$ to have rank $t+x$, and $\row(Y)$ to be the direct sum of $\row(X)$ and $\row(W)$. Indeed, Lemma~\ref{lem:subspace_intersection} shows that with high probability $\row(W)$ intersects the $(m-t)$ dimensional subspace $V$ trivially, where $V=\{(v_1,v_2,\ldots,v_m):v_i=0\text{ for }i\leq t\}$. Note that $\row(X)\subseteq V$. Row reduction does not change the row space, so the techniques in the proof of Lemma~\ref{lem:generic_dimensions} show that with high probability $Z$ has the form
\[
Z=\left(\begin{array}{c|c|c}
I_{t\times t}&\star&\star\\\hline
0_{x\times t}&I_x&U\\\hline
0_{(n-x-t)\times t}&0_{(n-x-t)\times x}&0_{(n-x-t)\times y}
\end{array}\right)
\]
where the first $t$ rows of $Z$ are a basis for $\row(W)$ modulo $\row(X)$ and the next $x$ rows are a basis for $\row(X)$.
(If it is not of this form, we declare a decoding failure.) It is now easy to recover the data $U$ from the matrix $Z$.

When $q$ is fixed with $n$, $m$ and $t$ large, the same coding method will work, but the dimensions of $U$ need to be reduced slightly, so that $Z$ is more likely to have the form above.

The coding scheme above restricts input matrices to having rank $x$, so we are essentially using the $k$-AMMC with $k=x$. In Section~\ref{sec:overview} we remarked that the $k$-AMMC may be asymptotically approximated by the operator channel from~\cite{KoetterKschischang} with $\rho$ erasures and $\tau$ errors, for certain integers $\rho$ and $\tau$. So an alternative approach would be to use a subspace code for encoding. As we remarked in Section~\ref{sec:overview}, this has the great advantage that decoding would also work for a range of practical channels where the number of errors is bounded by $t$ but not fixed. Of course a suitable subspace code must exist, but this is not an issue in practice as $q$ and $m$ are taken to be large.

\end{document}